\documentclass[10pt,conference]{IEEEtran}

\pdfoutput=1 
\usepackage{hyperref}
\usepackage{color,overpic,cite}
\usepackage{xcolor}
\usepackage{pdfsync}
\usepackage{url}
\usepackage{amsmath,amsfonts,amssymb,amsthm,bbm,bm,comment,multirow,subfigure}
\usepackage{tcolorbox}
\usepackage[ruled,commentsnumbered]{algorithm2e}  
\usepackage{siunitx}

\newcommand{\eat}[1]{{}}

\def\Y{\mathcal{Y}}
\def\X{\mathcal{X}}
\def\bE{\mathbb E}

\def\R{\mathcal{R}}

\newtheorem{lemma}{Lemma}

\newcommand\orangesout{\bgroup\markoverwith{\textcolor{orange}{\rule[0.5ex]{2pt}{0.4pt}}}\ULon} 

\usepackage{enumitem}

\usepackage{mdframed}

\usepackage{bm}

\SetAlFnt{\small}
 
\usepackage[ruled,commentsnumbered]{algorithm2e} 
\usepackage{tabularx,booktabs} 
\usepackage[normalem]{ulem} 

\title{Real-Time Scheduling of O-RAN Base Stations} 
\title{Real-Time Scheduling of Softwarized Base Stations} 
\title{Bandit-Learning Scheduling of Virtualized Base Stations} 
\title{Bandit-Learning Scheduling of Virtualized Base Stations} 
\title{Real-Time Scheduling of Softwarized Base Stations with Bandit Online Learning} 
\title{Bandit Learning Scheduling for Base Stations} 
\title{Bandit Learning for Base Station Scheduling} 
\title{Learning to Schedule Virtualized Base Stations} 
\title{Energy-aware Scheduling of Virtualized Base Stations in O-RAN with Online Learning}

\author{
    \IEEEauthorblockN{Michail Kalntis, George Iosifidis}
    \IEEEauthorblockA{Delft University of Technology, The Netherlands
    	\\
    Email: \{m.kalntis, g.iosifidis\}@tudelft.nl
}
}

\IEEEoverridecommandlockouts\IEEEpubid{\makebox[\columnwidth]{ 978-1-6654-3540-6/22\$31.00~\copyright~2022 European Union \hfill} \hspace{\columnsep}\makebox[\columnwidth]{ }}

\begin{document}

\maketitle

\begin{abstract}
The design of Open Radio Access Network (O-RAN) compliant systems for configuring the virtualized Base Stations (vBSs) is of paramount importance for network operators. This task is challenging since optimizing the vBS scheduling procedure requires knowledge of parameters, which are erratic and demanding to obtain in advance. In this paper, we propose an online learning algorithm for balancing the performance and energy consumption of a vBS. This algorithm provides performance guarantees under unforeseeable conditions, such as non-stationary traffic and network state, and is oblivious to the vBS operation profile. We study the problem in its most general form and we prove that the proposed technique achieves sub-linear regret (i.e., zero average optimality gap) even in a fast-changing environment. By using real-world data and various trace-driven evaluations, our findings indicate savings of up to 74.3\% in the power consumption of a vBS in comparison with state-of-the-art benchmarks. 
\end{abstract}
\smallskip


\section{Introduction}\label{sec:intro}

\subsection{Background \& Motivation} The importance of virtualizing the base stations is best manifested by the current flurry of industrial activities aiming to develop and standardize O-RAN architectures \cite{o-ran-andres}. The O-RAN Alliance is a global initiative that is devoted to revolutionizing Radio Access Networks (RAN). Its goal is to decentralize a field that has hitherto been dominated by a small number of companies and decrease the entry barrier for more potential firms. In this sophisticated system, RANs are constructed on virtualized network modules, resulting in virtualized Radio Access Networks. The focal point of these components are the virtualized Base Stations, which can be henceforward hosted on various devices, such as commodity servers or tiny embedded devices, and offer the possibility of significant operational/capital expenditure (OPEX/CAPEX) reductions. Promising examples include the open-source srsLTE \cite{gomez2016srslte} and OpenAirInterface (OAI) \cite{oai}.

Indeed, there is a wide consensus that the vBS' programmability can bring crucial performance gains and add the much-needed versatility to the otherwise-rigid, RAN systems. Alas, these benefits come at a cost. These softwarized base stations are found to have less predictable performance and more volatile energy consumption \cite{rost-globecom15}, \cite{jose-icc21}, \cite{ayala2021bayesian}, an effect that is amplified when instantiating them in general-purpose computing infrastructure. Hence, it is imperative to understand how to operate, or \emph{schedule}, these vBSs in order to unblock their wide adoption without raising the energy costs of mobile networks to unsustainable levels.

First works aiming at this direction focus, and rightfully so, on learning vBS \emph{meta-policies} (details below). These rules are decided at non-real-time scale and then imposed on the real-time schedulers of each vBS. O-RAN proposals have provisions for such two-level scheduling \cite{oran-spec-arch}, \cite{oran-spec-scenarios}; and include recommendations for employing learning tools to increase, e.g., the long-term throughput. Nevertheless, to bridge the gap between theoretical proposals and practical results, it is necessary to learn effective policies without relying on strong (and often unrealistic) assumptions such as knowing the full vBS operation profile, the expected data traffic, or future channel conditions. Otherwise, it is likely to be trapped in highly-suboptimal vBS operation points; a finding that we quantify here in terms of excess (i.e., unnecessary) energy costs that can add up to $74\%$ (details in Sec. \ref{sec:evaluation}).

\emph{The goal of this work} is to take the next step in this crucial problem by proposing and evaluating a robust algorithm that identifies effective meta-policies (or, simply \emph{policies} hereafter) and is oblivious to information about the underlying vBS operation, their hosting platforms, network conditions and data traffic (or, user needs). The core idea is to model the vBS scheduling as a \emph{bandit learning} problem \cite{bubeck_bandits} and design an algorithm that has provably-optimal performance under extensive conditions. The optimality criterion we employ is a combined objective of effective throughput (i.e., modulated by the users' traffic) and energy consumption, where the latter can be prioritized via a tunable weight parameter. Moreover, unlike prior works, our algorithm is lightweight and has minimal overheads, hence can be easily implemented in practice.

\subsection{Related Work} 
The idea of optimizing resource management in softwarized networks is not new, and prior works have focused mainly on \emph{(i)} models that relate control knobs to performance functions, \emph{(ii)} model-free approaches that rely on training data, and \emph{(iii)} Reinforcement Learning (RL) techniques. Model-based examples include \cite{rost-globecom15} and \cite{bega2018cares}, which maximize the served traffic subject to vBS computing capacity. However, vBS operation is heavily affected by the hosting platform and network conditions \cite{jose-icc21}, which renders such models impractical. Model-free approaches employ, e.g., Neural Networks, to approximate the performance functions of interest \cite{patras-DNN-Tutorial2019}, and have been used for network slicing \cite{bega-NNSlicing-JSAC2020}, edge computing \cite{jiajia-ML-netw2020}, etc. Yet, their efficacy is conditioned on the availability of training data. Another prominent approach focuses on runtime observations and is known as Reinforcement Learning. It is used, for example, in interference management \cite{alcaraz2020online} and the deployment of Software-Defined Networking (SDN) controllers \cite{poular-RLSDN-ICC19}. RL solutions, however, suffer from the curse of dimensionality and do not offer performance guarantees.

Following an akin approach, contextual bandit algorithms have been employed to decide video streaming rates \cite{bastian-cba-infocom19} or BS handover thresholds \cite{chuai-collab-INFOCOM19}; assign Central Processing Unit (CPU) time to virtualized BSs  \cite{vrain_conf}; and control millimeter Wave (mmWave) networks \cite{tekin-TCCN2020}. These works require \emph{context}-related information (e.g., about network conditions and traffic), which shapes the performance functions, to be known before the system is configured. More recently, versatile Bayesian learning techniques have been proposed for configuring vBSs \cite{ayala2021bayesian} \textemdash which is closer to the current work. Despite their promise, these solutions require knowing the context, as well as all system perturbations to be stationary over time. Nonetheless, these assumptions are restrictive, especially for heterogeneous and small-cell networks.


To overcome these obstacles, we follow a fundamentally different path and design a vBS control scheme that builds upon the seminal \emph{Exp3} algorithm \cite{exp3_auer}. Unlike all prior works, our approach: \emph{(i)} offers robust performance guarantees; \emph{(ii)} handles any type of network and load variations (even adversarial); \emph{(iii)} is oblivious to the (time-varying and unknown) vBS performance functions; and \emph{(iv)} exhibits low implementation complexity in terms of memory and computation requirements. This latter feature is in stark contrast with RL techniques (sizeable memory space required to store all space-actions combinations) and Bayesian approaches  \cite{freitas-tutorial-Proc2016} (heavy-duty matrix inversions). The proposed policy belongs to the class of \emph{adversarial} bandit learning, cf. \cite{bubeck_bandits}, which has been successfully used in network routing \cite{awerbuch-routing} and power control in Internet of Things (IoT) networks \cite{mert-iot}.

\subsection{Contributions} We design a learning algorithm that decides thresholds for key vBS operation knobs, namely for the vBS transmission power, the eligible Modulation and Coding Scheme (MCS), and the duty cycle (or airtime). The policy is updated at a near-real-time scale and is subsequently fed to the real-time schedulers that fine-tune the vBS parameters, see Fig. \ref{fig:architecture}. This type of meta-learning, i.e., deciding policies instead of fixing the vBS values directly, is central in the O-RAN architecture and has been recently proposed and studied experimentally, e.g., see \cite{vrain_conf, ayala2021bayesian} and references therein. The proposed algorithm relies on bandit feedback and makes no assumptions about how these knobs affect the vBS performance, nor assumes knowledge of the users' traffic during each scheduling period. This renders it practical for different types of vBS, hosting platforms, and network/traffic conditions. The main contributions of this paper are summarized below:
\begin{itemize}[leftmargin=4mm]
	\item We study the vBS scheduling problem in its most general form, i.e., in non-stationary adversarial conditions and without knowledge of traffic and vBS operation functions.		 
	
	\item We design an algorithm that achieves sublinear regret w.r.t. the (unknown) best vBS configuration and has minimal computation and memory overhead. This is the first work applying \emph{adversarial} bandit learning to vBS control.
	
	\item We use real-world traffic traces and testbed measurements to demonstrate the weaknesses of prior works \cite{ayala2021bayesian}, as well as the efficacy of the proposed learning algorithm.
	
	\item We release the source code\footnote{\url{https://github.com/MikeKalnt/BSvBS}} of our implementation online, under a permissive free software license, along with detailed documentation.
\end{itemize}

\textbf{Paper Organization}. The rest of this paper is organized as follows. Sec. \ref{sec:model} discusses the model and formally states the problem. Sec. \ref{sec:algorithm} introduces the bandit learning algorithm and Sec.  \ref{sec:evaluation} presents its data/trace-driven evaluation. Sec. \ref{sec:conclusions} concludes our study. 

\vspace{0.1cm}
\section{System Model and Problem Statement} \label{sec:model}

\textbf{O-RAN Background}. Our model follows the O-RAN proposals \cite{oran-spec-arch, oran-spec-scenarios, o-ran-andres}, which have provisions for learning-based resource management. ``Opening'' the RAN is a significant initiative that aims to expand the vRAN ecosystem. It enables multiple vendors to design components of the network architecture, which was previously monopolized by a confined number of large industries that provided end-to-end solutions. We consider a virtualized Base Station comprising a Baseband Unit (BBU) hosted by an off-the-shelf platform and being attached to a Radio Unit (RU). BBU corresponds to a Long-Term Evolution (LTE) eNodeB (eNB) for a 4G network and to a New Radio (NR) gNodeB (gNB) for a 5G network. For the latter, gNB is disaggregated into three focal components: (\emph{i}) the RU, (\emph{ii}) the DU, and (\emph{iii}) the CU.\footnote{In the O-RAN ecosystem, the terms O-CU, O-DU, O-RU, and O-eNB are used to denote the CU, DU, RU, and eNB, respectively.} The architecture of the implemented system can be seen in Fig. \ref{fig:architecture}. 

\emph{Our goal is to design performance/energy-optimizing configuration policies that adapt to network conditions and user needs}. We assume that an rApp (Policy Decider - PD), see Fig. \ref{fig:architecture}, is instantiated at the Non-RT RIC and implements an algorithm that learns to select efficient \emph{radio policies}. These are essentially \emph{adaptive threshold rules} which guide the underlying real-time schedulers towards the desirable vBS operation. The policy is communicated via the R1 interface to the Non-RT framework, and from there, it is provided to the Near-RT RIC via the A1 interface. In the Near-RT RIC, an xApp (Policy Enforcer - PE) forwards the radio policy to the E2 nodes\footnote{E2 nodes refer to RAN nodes, such as an O-CU, an O-DU, or an O-eNB.} through the E2 interface. The optimal policy depends on the network conditions and users' load, both of which may vary arbitrarily across time and are typically unknown when the policy is decided. At the end of each decision period $t$, the Near-RT RIC's Data Monitor computes a \emph{reward} by aggregating the adopted performance and energy cost metrics received via the E2 and feeds them to the PD through the O1 interface, before the next decision round.

\begin{figure}[t]
	\centering
	\includegraphics[scale=0.35]{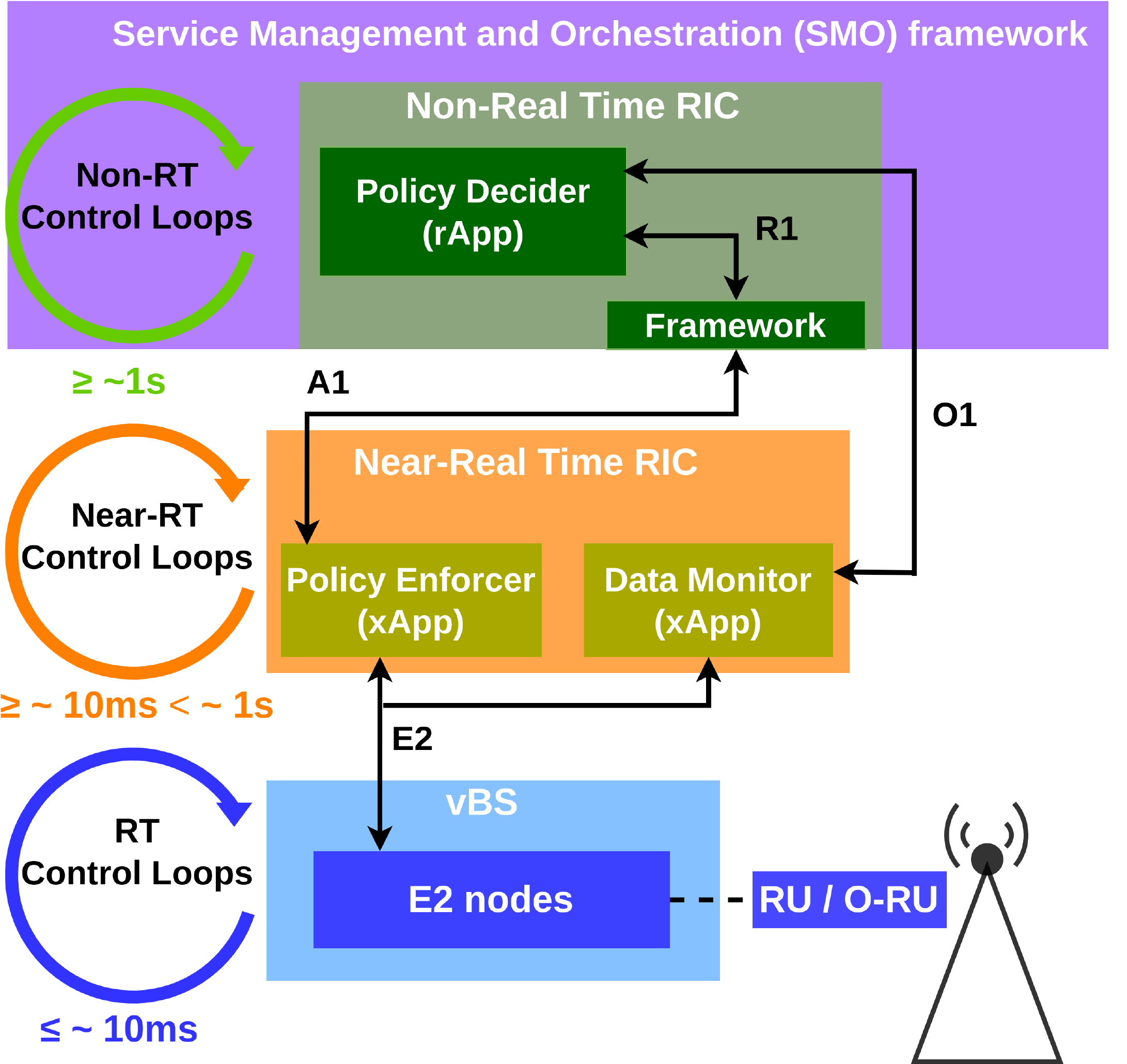}   
	\caption{\textbf{O-RAN-compliant  architecture \& workflow}. The key building block is the Non-Real-Time (Non-RT) RAN Intelligent Controller (RIC), hosted by the Service Management and Orchestration (SMO) framework; and the Near-Real-Time (Near-RT) RAN Intelligent Controller (RIC) \cite{oran-spec-scenarios}. The system has three control loops: (\emph{i}) Non-RT in the Non-RT RIC, which involves large-timescale operations with execution time $\!>\!1$sec, (\emph{ii}) Near-RT in the Near-RT RIC ($>\!10$msec), and (\emph{iii}) RT (Real-Time) control loops in the E2 nodes ($<\!10$msec).}
	\label{fig:architecture}
\end{figure}

\textbf{vBS Controls}. We consider time-slotted system operation in alignment with O-RAN specs, where each slot represents a period (range of a few seconds) over which a certain policy is being applied. We optimize the system operation over a time horizon of $t=1,\ldots, T$ slots, where $T$ can take arbitrarily large values and is decided in advance.\footnote{The assumption of fixing $T$ can be dropped by employing the doubling trick or time-adaptive learning parameters; details in next section.} Without loss of generality (WLOG), we assume unitary slot length.

Our policy includes thresholds for specific scheduling controls that are key to vBS performance, in line with recent measurement-based studies \cite{ayala2021bayesian}, \cite{jose-icc21}, \cite{vrain_conf}. In detail, for the downlink (DL) operation, we define the set of maximum allowed vBS transmission power control (TPC) $\mathcal{P}_{d} = \{p_{i}^{d}, \,\,\forall i \in [H] \}$, the set of highest eligible MCS $\mathcal{M}_{d} = \{m_{i}^{d}, \,\,\forall i \in [I] \}$ and the set of maximum vBS transmission airtime, or duty cycle $\mathcal{A}_{d} = \{a_{i}^{d}, \,\,\forall i \in [J] \}$, where $H$, $I$, and $J$ denote the number of transmission power, MCS, and airtime levels in DL, respectively. Hence, in period $t$, we determine the DL control:

\[x_t^{d} \in \mathcal{P}_{d} \times \mathcal{M}_{d} \times \mathcal{A}_{d}.
\]

For the uplink (UL) operation, we introduce the set $\mathcal{M}_{u} = \{m_{i}^{u}, \,\,\forall i \in [K] \}$ and $\mathcal{A}_{u} = \{a_{i}^{u}, \,\,\forall i \in [L] \}$, where $K$ and $L$ express the number of MCS and airtime levels in UL, respectively.\footnote{A UL TPC policy is not defined since the users' transmission power has less impact on the vBS power than the MCS and UL airtime.} Thus, the UL control in period $t$ is: 

\[x_t^{u} \in \mathcal{M}_{u}\! \times\! \mathcal{A}_{u}.\] 

Following 3GPP specs, we assume these controls take values from a finite set that includes all possible combinations:\footnote{For instance, the MCS values are predetermined, and similarly, one can quantize the power and airtime values; see, e.g. \cite{gomez2016srslte} for the \emph{srsLTE} vBS.}
\[
\mathcal{X} = \mathcal{P}_{d} \!\times\! \mathcal{M}_{d}\! \times\! \mathcal{A}_{d}\! \times\! \mathcal{M}_{u}\! \times\! \mathcal{A}_{u}.
\]
{Thus, the \emph{radio policy} in period $t$ is specified as:
\nolinebreak
\[x_t = (x_t^{d}, x_t^{u}) \in \X.\]}

\textbf{Rewards \& Costs}. The first goal of the learner is to maximize the \emph{effective} DL and UL throughput, which depends on the aggregate of the transmitted data and the backlog in each direction. In particular, in line with prior works (see \cite{ayala2021bayesian} and references therein), we use the following {\emph{utility function}}:
\begin{equation} \label{eq:utility}
	U_t(x_t)=
	\log \left(1\!+\!\frac{R_t^{d}(x_t^{d})}{d_t^{d}}\right)\!+\log \left(1\!+\!\frac{R_t^{u}( x_t^{u})}{d_t^{u}}\right),
\end{equation}
when $d_t^{d} \!> \!0$ and $ d_t^{u} \!>\! 0$, with $U_t(x_t) \!=\! 0$ otherwise.
$R_t^{d}(\cdot)$ and $R_t^{u}(\cdot)$ denote the DL and UL transmitted data during period $t$; and $d_t^d$ and $d_t^u$ are the respective backlogs (i.e., the user needs/requests during $t$). The logarithmic transformation balances the system utility across each stream, but we note that other mappings (e.g., linear) are eligible.

The second goal of the policy is to minimize the vBS energy cost. To that end, we introduce the \emph{power cost} function $P_t(x_t)$, which depends on the configuration $x_t$ in an unknown and possibly time-varying fashion. We kindly refer the reader to our experimental study \cite{jose-icc21} regarding the challenges in modeling the vBS power cost. We focus principally on \emph{(i)} the power consumption of the CPU at a vBS, which has the lion's share of the total power consumed at the BBU \cite{jose-icc21}; and \emph{(ii)} the total vBS power consumption (including the RU). We consider these distinct cases to capture the scenarios (arising in practice) where the DU (hosting the BBU) and RU have, or do not have, a common power source. Therefore, we model $P_t(\cdot)$ as a black-box with values observed in runtime.

Putting the above together, the performance criterion for the {PD} is the \emph{reward function} $ \tilde f_t\!: \mathbb \X \rightarrow \mathbb R$ defined as:
\begin{equation}\label{eq:reward}
\tilde f_t(x_t)={U_t(x_t)}-\delta {P_t(x_t)},
\end{equation}
where the parameter $\delta\!>\!0$ is set by the {PD} to tune the relative priority of the {utilities} and energy costs. It also serves as a metric transformation, allowing a meaningful scalarization of the function components. Further to that, we introduce, for technical reasons, the \emph{scaled} reward function $f_t:\X \rightarrow [0,1]$ since our learning algorithm operates on that interval. An easy-to-implement mapping that ensures this normalization is:
\begin{equation}
f_t(x_t)=\big( \tilde f_t(x_t)- \tilde f_{min} \big)/ \big(\tilde f_{max}-\tilde f_{min}\big).  \label{eq:mapping}
\end{equation}
The scaling parameters $\tilde f_{min}$ and $\tilde f_{max}$ can be directly determined beforehand based on the value of $\delta$ (depending on the importance given to each component of the function), the minimum/maximum value of the power cost function (might be the minimum and maximum value of the monetary cost associated with the vBS operation), the minimum/maximum transmission power of the vBS, airtime, MCS and user loads. The bounded reward assumption comes WLOG.


\textbf{Environment \& System Volatility}. It is crucial to note that both reward components vary with time. There are several factors contributing to this effect. First, the user traffic that shapes $U_t$ changes, sometimes drastically \textemdash e.g., in small cell networks where user churn is high. Second, the network conditions might as well vary (in slow, fast, or mixed timescales), and this affects the achieved data transmissions (hence $U_t$ changes even for fixed $x_t$) but also impacts the energy cost $P_t$ (low Signal-to-Noise Ratio (SNR) induces more BBU processing costs \cite{jose-icc21}). Third, the operation cost of the vBS hosting platform is subject to the variations of external computing loads (e.g., when co-hosting other services or other vBS/DUs), changes in the monetary cost (or availability) of the energy price, and so on. Importantly, all the above factors are \emph{unknown} at the beginning of each scheduling period $t$. Indeed, it is challenging to predict the user loads, energy availability, channel conditions, etc., over a few seconds. This, in turn, means that often in practice, when we decide $x_t$ in each slot, we do not have access to the function $f_t$.

\textbf{Learning Objective}. The goal of the {PD} is to find a sequence of configurations $\{x_t\}_{t=1}^T$ that aggregate {rewards} so as to approach, asymptotically, the {cumulative reward} achieved by the single best ({ideal}) configuration. Formally, we employ the metric of \emph{static expected regret}:
\begin{equation}
\R_T= \max_{x\in \X	} \left\{ \sum_{t=1}^T f_t(x)\right\} - \bE \left[ \sum_{t=1}^T f_t(x_t)	\right], \label{def:regret}
\end{equation}
\noindent where the first term describes the best configuration that can be only selected with hindsight, i.e., with a priori knowledge of all future {reward} functions until $T$; and the second term measures the achieved {cumulative reward} by our policy. Note that the expectation is induced by any possible randomization in the selection of $\{x_t\}$ that is introduced by the {PD}.

Eventually, our objective is to devise a rule that decides the configurations in such a way that the average regret,\footnote{Due to the exploration phase of the algorithm proposed, it might happen that $\R_{t+1} > \R_t$ because the algorithm in slot $t+1$ explored a configuration that was under-performing compared to the configuration chosen in slot $t$. However, the average regret could still diminish to zero as time evolves. More information is available in Sec. \ref{sec:algorithm}.} for any possible realization of {rewards} $\{f_t\}_{t=1}^T$, diminishes asymptotically to zero, i.e., $\lim_{T\rightarrow \infty}\R_T/T=0$. Furthermore, we wish to ensure this condition without knowing the sequence of {rewards}, not even having access to $f_t$ at the time $x_t$ is being decided. This makes our policy applicable to a range of practical scenarios, such as in highly volatile wireless environments, locations with high user churn, or small-cells where user demands are non-stationary  \cite{paschos-comm-mag}.

\vspace{0.1cm}
\section{Bandit Learning Algorithm} \label{sec:algorithm}

Our PD builds on the seminal \emph{Exp3} algorithm \cite{exp3_auer} that imposes no assumptions on the sequence of {rewards} $\{f_t\}_{t=1}^T$.  This, in turn, means that the obtained performance bounds are guaranteed to hold independently of how the network and environmental factors vary and/or affect the rewards and costs. Hence, unlike prior works such as \cite{ayala2021bayesian} and \cite{vrain_conf}, the achieved performance is robust and valid for non-stationary conditions. Besides, as we demonstrate in Sec. \ref{sec:evaluation}, the actual performance is often substantially higher than these worst-case bounds. Another prominent feature of our  algorithm is that it works with \emph{bandit} feedback, i.e., it is adequate to observe the outcome $f_t(x_t)$  of the employed configuration $x_t$ \textemdash instead of the entire $f_t(\cdot)$, which remains unknown.\footnote{Clearly, when applying $x_t$ we observe $f_t(x_t)$ but do not learn how the system would have performed for any other configuration $x\!\in\! \X$, with $x\!\neq\!x_t$.}

In detail, the underlying idea is that we learn, on the fly, the correct probability distribution $y_t$ (the sequence of distributions) from which we can draw the configuration $x_t$ for each period $t$:
\[
x_t \sim \mathbb{P}(x_t=x') = y_t(x'), \forall x' \in \mathcal{X}.
\] 
The distributions $\{y_t\}_t$ belong to the probability simplex:
\[
 \Y=\left\{ y\in [0,1]^{|\mathcal{X}|} \,\, \bigg \vert \,\,\sum_{x\in \mathcal{X}} y(x)=1	\right\}, \notag
\]
and are calculated in each round using the following carefully-crafted explore/exploit rule:
\begin{equation}
		y_t(x) = \frac{\gamma}{|\mathcal{X}|} + (1-\gamma)\frac{w_t(x)}{\sum_{x}{w_t(x)}},\,\,\, \forall \, x\in \mathcal{X}. \label{eq:update-distribution}
\end{equation}
Parameter $\gamma\! \in\! [0,1]$ determines the extent to which the PD samples a configuration randomly (exploration), or prioritizes those configurations found to perform well (exploitation). The latter  happens with the help of the weight vector $w_t\!=\!\big(w_t(x)\!:\! x\!=\!1,\ldots, |\mathcal{X}|\big)$ that tracks the success of each tested policy. In particular, we update the weights at the end of each period, using the exponential update:
\begin{equation}
	w_{t+1}(x)=w_t(x)\exp \left(\frac{\gamma \Phi_t(x)}{|\mathcal{X}|} \right),\,\,\,\forall\,x\in\mathcal{X}, \label{eq:weights}
\end{equation}
where $\Phi_t(x)$ is an unbiased function estimator defined as:
\begin{equation}
\Phi_t(x)=
\begin{cases}
	f_t(x_{t}) / y_t(x_t), & \text{if} \,\, x = x_t,\\
	0, & \text{otherwise.}
\end{cases}
\label{eq:weighted-feedback}
\end{equation}
 {We recall that $x_t$ is the selected configuration in slot $t$. If the estimator $\Phi_t(x)$ is used to estimate the actual reward $f_t(x_t)$, it is straightforward to see that:
\[
    \bE[\Phi_t(x) | x_1,x_2, \ldots, x_{t-1}]=f_t(x_t), \notag
\]
where $x_1,\ldots, x_{t-1}$ are the configurations chosen up to ${t-1}$. By weighting each observed value with its selection probability, we ensure that the PD will eventually explore configurations with a small probability.}

The steps of the proposed learning scheme are summarized in Algorithm \ref{alg1}, which takes as input the time horizon $T$ over which we optimize the vBS operation and the number of eligible configurations $|\mathcal{X}|$; this information suffices to optimize the value of the exploration parameter $\gamma$. The performance of the algorithm is summarized in the following lemma, which holds for any possible sequence of functions $\{f_t\}_{t=1}^T$:
\begin{lemma}\label{corrolary1}
Algorithm 1 for a fixed horizon $T$ ensures expected regret defined in \eqref{def:regret}:
\begin{equation}
	\R_T \leq 2\sqrt{(e - 1)}\sqrt{T|\mathcal{X}|\ln{|\mathcal{X}|}} \label{eq:exp3_bound}
\end{equation}
\end{lemma}
\begin{proof}
The proof follows by tailoring the main result of \cite{exp3_auer}. We provide a brief but sufficient explanation. In particular, for selecting $\gamma$, we need to determine an upper bound $g$ on the cumulative reward of the best configuration until $T$. Given that: \emph{(i)} the horizon $T$ is known in advance; and \emph{(ii)} the rewards $f_t(x_t)$ for each chosen configuration $x_t$ at time $t$ cannot be greater than 1 (due to the proposed normalization described in Sec. \ref{sec:model}), the value of $g$ can be set equal to $T$, i.e., $g \!=\! T$. Also, the number of \emph{bandit arms} in our case corresponds to the eligible configurations; hence it is equal to $|\X|$.
\end{proof}

\setlength{\textfloatsep}{0pt}
\begin{algorithm}[t] 
	\nl \textbf{Input}: Horizon $T$; Configurations $|\mathcal{X}|$;\\ 
	\nl \textbf{Initialize}: $\gamma = \min\left\{1, \sqrt{\dfrac{|\mathcal{X}| \ln{|\mathcal{X}|}}{(e-1)T}}\right\}$; \\[2mm] \hspace{1.35cm}{$w_{1}(x)\leftarrow 1, \,\,\forall x\in \mathcal X$.}\\[1mm]
	\nl \For{ $t=1,2,\ldots, T$  }{
		\nl Update the distribution using \eqref{eq:update-distribution}.\\
		\nl Sample next configuration: $x_t \sim y_t$.\\
		\nl Receive \& scale reward $f_t(x_t)$.\\
		\nl Calculate weighted feedback using \eqref{eq:weighted-feedback}.\\
		\nl Update the weights using \eqref{eq:weights}. 
	}
	\caption{{Bandit Scheduling for vBS (BSvBS)}}\label{alg1}
	\setlength{\intextsep}{0pt} 
\end{algorithm}

We notice that the expected regret is indeed sublinear $\R_T=o(T)$, which ensures that its time average diminishes to zero. Hence, Algorithm \ref{alg1} is guaranteed to achieve the same performance as the (unknown) {single best configuration}, without imposing any conditions on the system operation, network conditions, or user demands. This robust behavior fills the gap of recent related works \cite{vrain_conf}, \cite{ayala2021bayesian}. Moreover, we highlight that the regret depends on the number of possible meta-policies up to a square root factor. And while their number is expected to be smaller than the number of policies applied to the RT O-RAN level, this finding still points to an interesting direction for further reducing this dependency. 

\vspace{0.1cm}
\section{Performance Evaluation} \label{sec:evaluation}

\textbf{Experimental Setup \& Scenarios}. We evaluate Algorithm 1 in a variety of scenarios using our recent dataset \cite{ayala2021bayesian}, which includes measurements of the power consumption and performance of vBS policies. The experiments have been conducted using a  \emph{srsRadio} vBS \cite{gomez2016srslte}, and we have used its default schedulers for the underlying real-time decisions (which comply with our meta-policies).\footnote{The dataset contains \SI{32797}{} measurements for different policies, fixed for approximately one minute. The experiments are carried out on a small factor general-purpose PC (Intel NUC BOXNUC8I7BEH with CPU i7-8559U@\SI{2.70}{\giga\hertz}), which deploys the BBU and is configured with a bandwidth of \SI{10}{\mega\hertz}. This means that it supplies a maximum capacity of approximately \SI{32}{Mbps} and \SI{23}{Mbps} for the downlink and uplink operation, respectively. See \cite{ayala2021bayesian} for details.} 

The random perturbations in this setup emanate due to time-varying UL and DL demands, $\{d_t^{u},d_t^{d}\}_t$, and time-varying CQIs (Channel Quality Indicators), $\{c_t^{u}, c_t^{d}\}_t$. The latter are integer numbers sent from the User Equipment (UE) to the network to designate how good or bad the channel quality is. The dataset contains $|\mathcal{X}| \!=\! 1080$ configurations (policies), but we use a subset of them because calculating the best configuration in hindsight is computationally challenging when $|\mathcal{X}|$ is large.\footnote{We stress that this benchmark configuration is needed for the plots of $R_T$, but \textbf{it is not required} when one uses the algorithm in practice. Hence, this limitation is related only to presenting the regret here.} For the power cost function, we set $P_t(x_t) \!=\! V_t$, where $V_t$ is the total power consumed by the vBS. We also set $\delta=1$ to prioritize the minimization of the power consumption.


\begin{figure}[t]
	\centering
	\subfigure[]{\label{fig:nonstat-BP-vRAN}\includegraphics[scale=0.35]{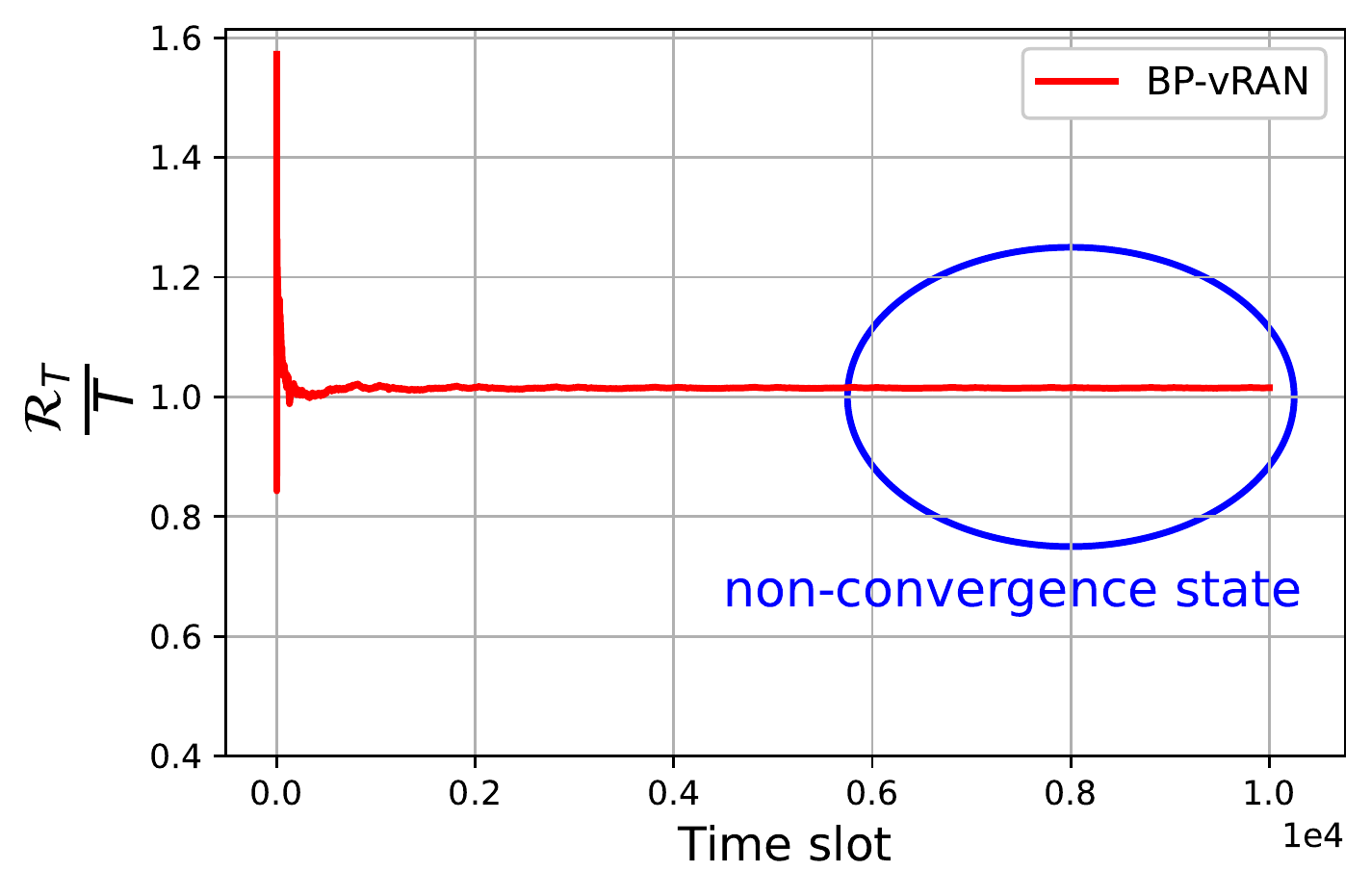}} 
	\subfigure[]{\label{fig:chosen_configurations}\includegraphics[scale=0.27]{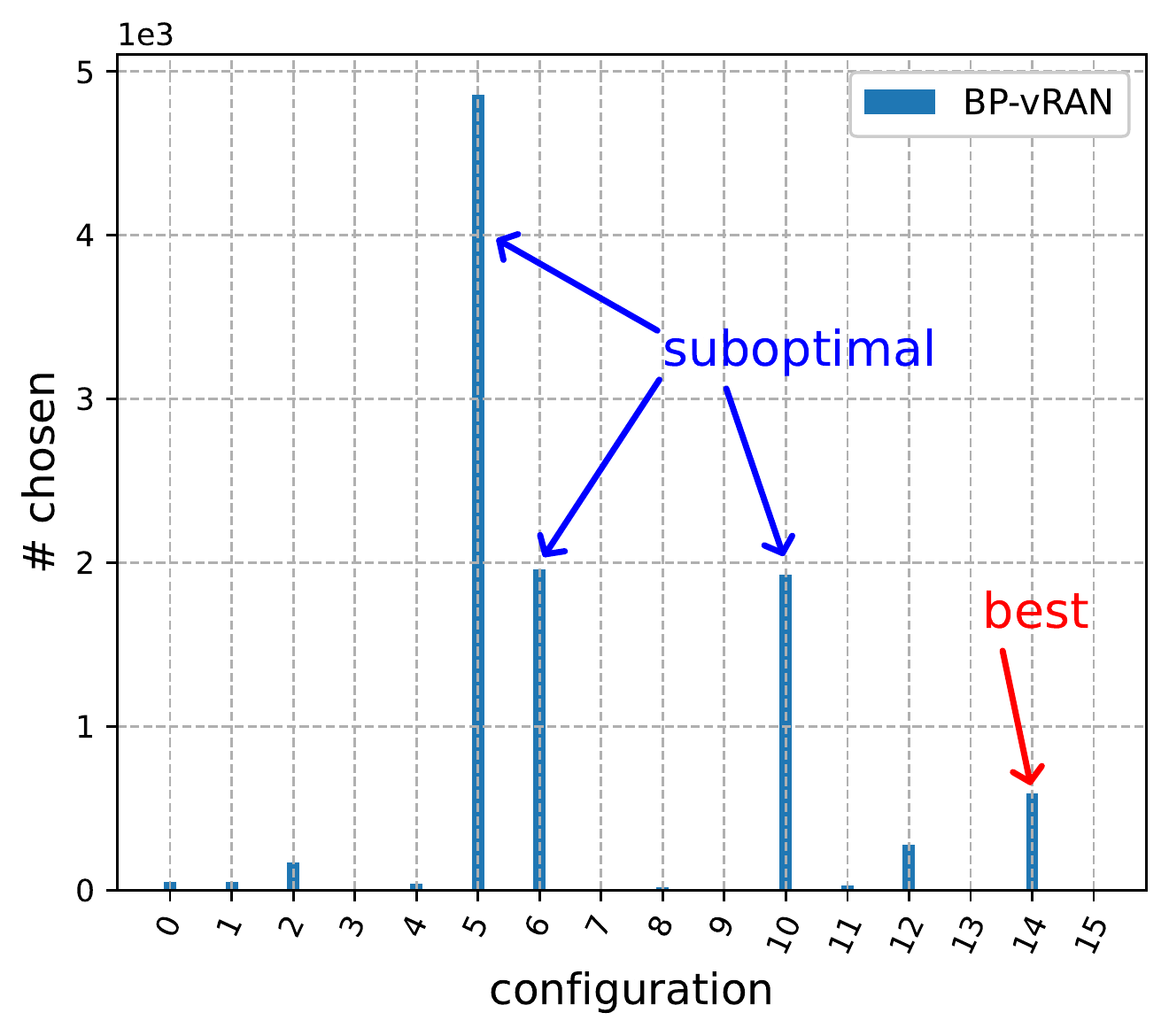}}
	\caption{BP-vRAN regret in Scenario B for $|\mathcal{X}|\!=\!16$. \textbf{(a)}: Evolution of $R_T/T$ in non-stationary conditions. \textbf{(b)}: Numbers of times each eligible policy was chosen in a window of $T=10k$ slots. }
	\label{fig:BP-vRAN_T-10000}
    \medskip
\end{figure}

For the following analysis, we assume that the traffic loads discerned in our system are sampled, either from $d_t^{d} \sim \mathcal{U}(29, 32)$, $d_t^{u} \sim \mathcal{U}(20, 23)$ (high DL and UL demands, respectively), or from $d_t^{d},\, d_t^{u} \sim \mathcal{U}(0.01, 1)$ (low DL and UL demands, respectively). Similarly, the channel qualities are drawn, either from $c_t^{d},\, c_t^{u} \sim \mathcal{U}(13, 15)$ (good channel qualities\footnote{CQI 13 and 15 correspond to SNR of \SI{25}{\decibel} and \SI{29}{\decibel}, respectively.} in DL and UL, respectively), or from $c_t^{d}, c_t^{u} \sim \mathcal{U}(1, 3)$ (poor channel qualities\footnote{CQI 1 and 3 correspond to SNR of \SI{1.95}{\decibel} and \SI{6}{\decibel}, respectively.} in DL and UL, respectively). According to these distributions, we discern two scenarios:
\begin{itemize}[leftmargin=4mm]
    \item \underline{Scenario A}: the demands and CQIs are consistently drawn uniformly at random from the \emph{high} distribution (\emph{stationary}).
    
    \item \underline{Scenario B}: the demands and CQIs are drawn randomly from the \emph{high} distribution in slots $\{2t-1\}_{t=1}^{\lceil T/2 \rceil}$ and from the \emph{low} distribution in slots $\{2t\}_{t=1}^{\lfloor T/2 \rfloor}$. 
\end{itemize}

The first scenario aligns with the experiments in recent studies, e.g., \cite{ayala2021bayesian}, \cite{jose-icc21}, \cite{vrain_conf}. The second one implements a \emph{ping-pong} strategy that corresponds to the most challenging-to-learn \textbf{adversarial} scenario in regret analysis, cf. \cite{bubeck_bandits}. Clearly, an algorithm that performs well under this case is expected (in fact, guaranteed) to perform well in all other scenarios.

\textbf{Gap in Prior work}. We first show that state-of-the-art works under-perform in commonly-encountered non-stationary conditions. We focus on BP-vRAN \cite{ayala2021bayesian}, which relies on the seminal GP-UCB algorithm, cf. \cite{freitas-tutorial-Proc2016}. To demonstrate that even a simple case hampers its operation, we focus on a subset of configurations: $|\mathcal{M}_{d}|\!=\! |\mathcal{M}_{u}|\! =\! |\mathcal{A}_{u}| \!=\! |\mathcal{A}_{d}| \!=\! 2$ and $|\mathcal{P}_{d}| \!=\! 1$, thus $|\mathcal{X}| \!=\! 16$. 

BP-vRAN models the user demands and CQIs as \emph{context}, which are observed before a policy is decided. Given that the context directly impacts the selection of the controls, we will show how changes in the network conditions and demand will affect the algorithm's success. We present an example where the context differs between its observation and application to the system. This case might apply quite often in practice, given that the slots of reference are of several seconds. For the plots in this section, it is crucial to note that the reward function $f_t(x_t)$ is unbounded.\footnote{For that reason, we avoid showing the expected regret of BP-vRAN and BSvBS in the same plots.}

We perform the experiments in Scenario B. Thus, even though we detect low demands and CQIs in the first slot, when the context is applied to the system, these values have increased significantly because, for example, many users with minimum noise interference arrived. In the next time slot, we observe high demands and CQIs, but the actual context has decreased when applied to the system because, for example, few users are present and SNR is reduced. According to the distributions mentioned above, the pattern remains the same in the following slots, i.e., context altering between high and low values in the observation and application to the system.

In Fig. \ref{fig:nonstat-BP-vRAN}, we underline that the expected average regret does not decrease, even after $T=10k$ slots. This happens because the algorithm takes decisions in each slot $t$ by assuming knowledge of $f_t$, which might take arbitrarily low or high values, depending on the network conditions. However, due to the system's volatility, the policy for each $t$ should be selected solely based on past values $\{f_\tau(x_\tau)\}_{\tau=1}^{t-1}$. Fig. \ref{fig:chosen_configurations} demonstrates that BP-vRAN insists on selecting configuration $\#5$ (approximately for $\SI{50}{\percent}$ of the slots) and picks the optimal configuration $\#14$ only for $590$ out of the $10k$ slots. This manifests its inability to explore the configuration space due to the non-stationary demands and CQIs. 

\textbf{Evaluation of Algorithm \ref{alg1}}. We consider both Scenarios A and B in our experiments, and we set $|\mathcal{X}|\!=\!256$ for the reasons mentioned above. 
%
%
Fig. \ref{fig:BSvBS} depicts the expected regret for \SI{20}{} independent runs in a window of $T\!=\!50k$ slots. More precisely, Fig. \ref{fig:BSvBS_1} plots the decay of $R_T$ for BSvBS in the stationary and adversarial cases. During the first $12k$ slots, the incurred regret for Scenario A is higher than the case of Scenario B, i.e., we perceive a \SI{7.5}{\percent} difference in slot $1k$. As time evolves and confidence in the performance of configurations is built, it is reasonable to observe the regret of the stationary case (Scenario A) to be lower than the more volatile system of Scenario B. That is, in slot $50k$, the regret for the adversarial case is \SI{12.7}{\percent} greater than the stationary's Scenario in slot $50k$. Furthermore, we see in Fig. \ref{fig:BSvBS_2} that the experienced regret is by far lower than the upper's bound; that is \SI{80.9}{\percent} and \SI{78.1}{\percent} lower for Scenario A and B, respectively.


\begin{figure}[t]
    \centering
    \subfigure[Evolution over time.]{\label{fig:BSvBS_1}\includegraphics[scale=0.3]{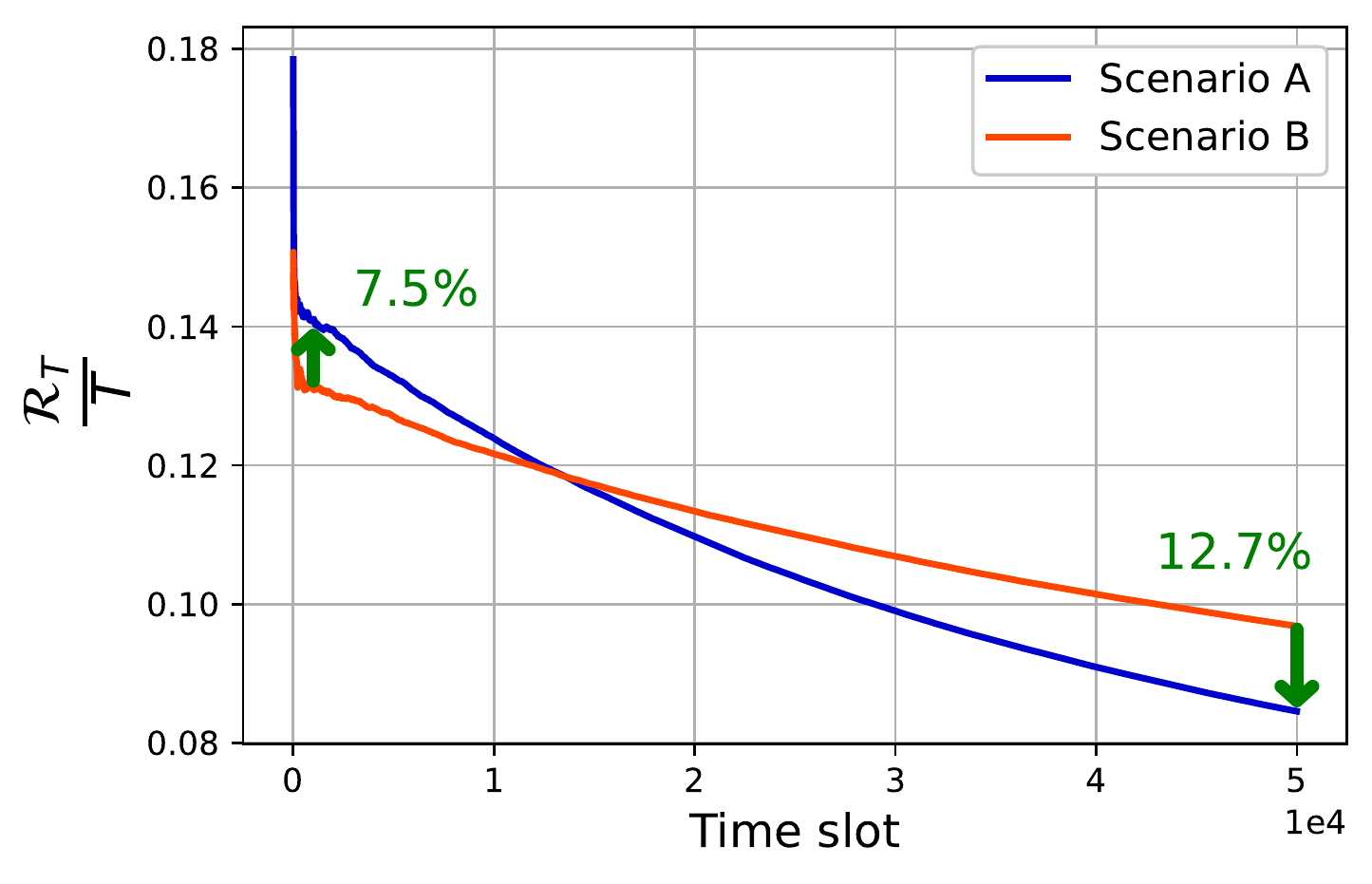}} 
    \subfigure[Upper bound.]{\label{fig:BSvBS_2}\includegraphics[scale=0.3]{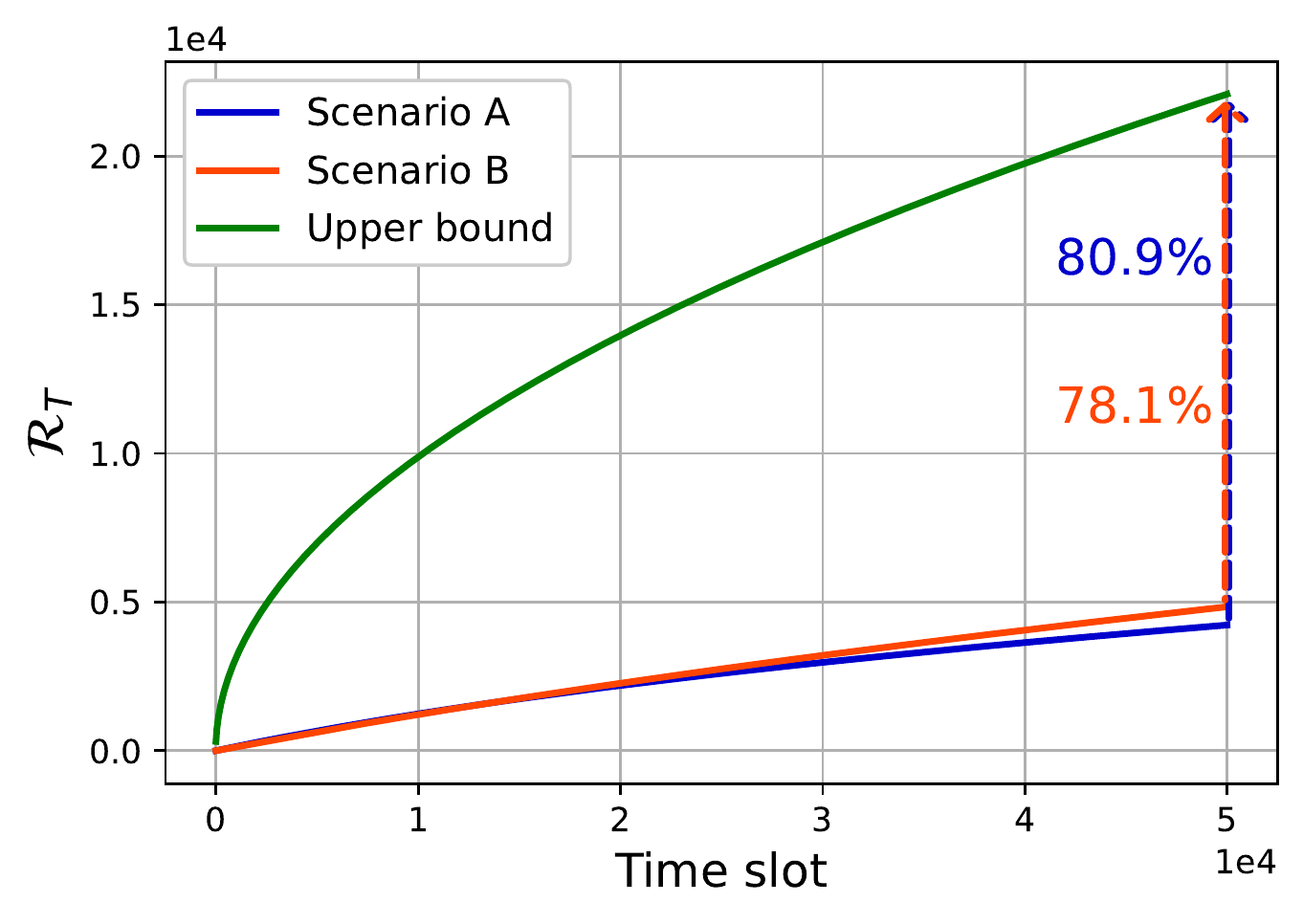}}
    \caption{BSvBS regret in Scenarios A and B for $|\mathcal{X}| = 256$.}
    \label{fig:BSvBS}
\end{figure}

\begin{figure}[t]
    \centering
    \subfigure[CPU power.]{\label{fig:energy_cpu}\includegraphics[scale=0.3]{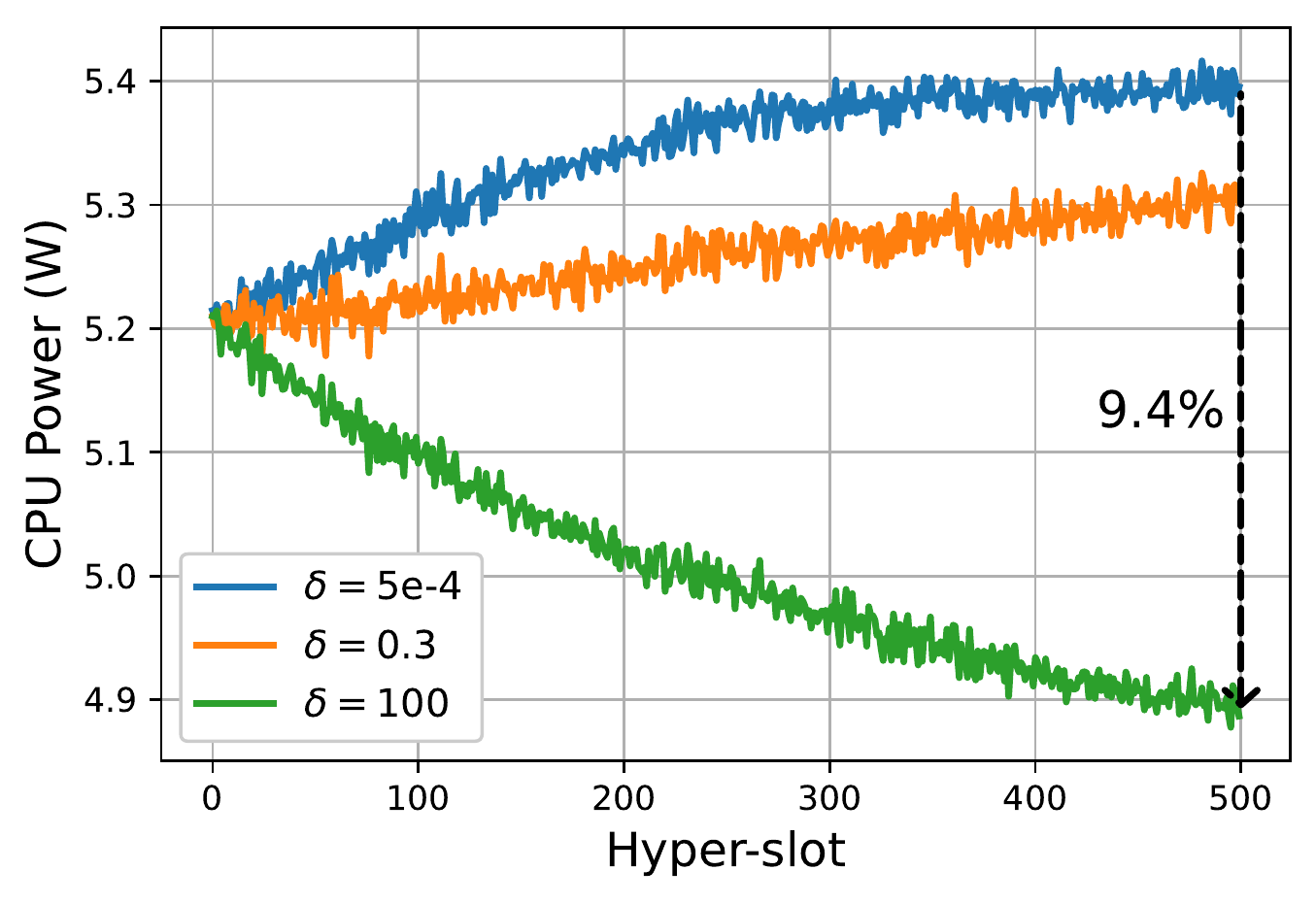}} 
    \subfigure[Total power.]{\label{fig:energy_total}\includegraphics[scale=0.3]{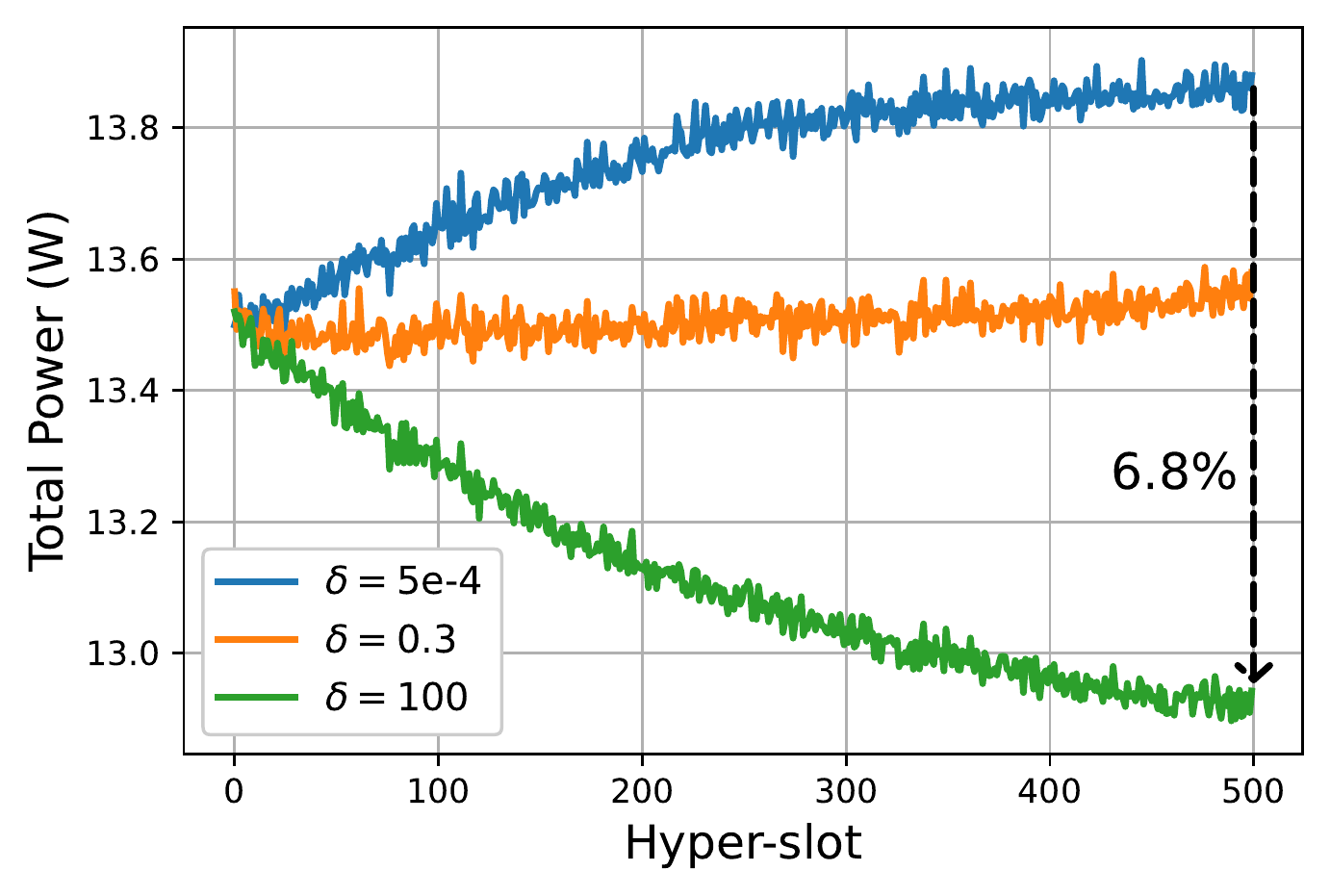}}
    \caption{BSvBS power consumption in Scenario B for $|\mathcal{X}| = 256$.}
    \label{fig:BSvBS_energy}
\end{figure}

In Fig. \ref{fig:BSvBS_energy}, we show the effect of $\delta$ in the consumed power. In detail, we plot the evolution of the power consumption in Scenario B and we distinguish two cases: \emph{(i)} the total power consumption; and \emph{(ii)} the BBU/CPU power consumption.\footnote{The testbed uses a single energy source, so the total power consumption aggregates the power of the entire platform and the radio component.} We run the algorithm for $T\!=\!100k$ slots, and we define a \emph{hyper-slot} of 200-slots-length to facilitate the presentation of results. Indeed, as $\delta$ increases, priority is given to the minimization of the power instead of the maximization of the utility (see eq. \ref{eq:reward}). Thus, for the hyper-slot \SI{500}{}, we manage to save \SI{6.8}{\percent} and \SI{9.4}{\percent} in the total and CPU power consumption, respectively, by using $\delta\!=\!100$ instead of $\delta\!=\!5e\!-\!4$. The operators of the system can use our proposal for conserving energy.

\begin{table}[t]
	\centering
	\caption{Power cost (Scenario B)} 
	\label{tab:power}
	\resizebox{\columnwidth}{!}{
		\begin{tabular}{cccccccccc}
			\toprule
			\multirow{2}{*}{Slots} & \multicolumn{3}{c}{CPU (kW)} & \multirow{2}{*}{\textbf{Saving} (\%)} & \multicolumn{3}{c}{Total (kW)} & \multirow{2}{*}{\textbf{Saving} (\%)} \\
			& {BP-vRAN} & {BSvBS} & {Min} & & {BP-vRAN} & {BSvBS} & {Min} \\
			\midrule
			$200k$ & 1052.6 & 992.1 & 955.1 & \textbf{62.1} & 2735.6 & 2609.8 & 2566.2 & \textbf{74.3}\\[0.2em]
			$100k$ & 534.6 & 501.1 & 476.7 & \textbf{57.9} & 1375.2 & 1313.2 & 1284.7 & \textbf{68.6} \\[0.2em]
			$50k$ & 262.4 & 252.9 & 235.6 & \textbf{35.5} & 677.3 & 660.9 & 635.0 & \textbf{38.8}\\
			\bottomrule\bottomrule
		\end{tabular}
	}
	\bigskip
\end{table}

Table \ref{tab:power} presents the power consumed by BP-vRAN and BSvBS in Scenario B and compares the gains each of them achieves w.r.t. the ideal-minimum-energy of the benchmark. We set $\delta=100$ to prioritize the power costs strongly over the accomplished utility. By running BSvBS for $T=50k$ slots, we get \SI{35.5}{\percent} and \SI{38.8}{\percent} savings in the CPU and total power, respectively, compared to BPvRAN. In other words, BSvBS approaches closer to the minimum possible energy cost than its competitor.\footnote{It is worth noting that vBS consumes considerable power even under the best configuration (Min column) or even when no user is active; thus, the direct comparison of the power would not demonstrate the actual gains.} Interestingly, we also see an increase in these attained gains by observing the performance of BSvBS for more time slots. That is, by doubling the number of slots, the savings increase by more than $1.5\times$ up to \SI{57.9}{\percent} and \SI{68.6}{\percent} in the CPU and total consumption;  and become \SI{62.1}{\percent} and \SI{74.3}{\percent}, respectively, for $T=200k$ slots. On the one hand, we note that BSvBS performs better as time passes and reaches closer to the consumption of the best configuration, even in the adversarial scenario. On the other hand, the indication that the power consumption of BP-vRAN almost doubles as the slots are doubled means that the algorithm is stuck exploiting under-performing configurations.

\vspace{0.1cm}
\section{Conclusions and Future Work} \label{sec:conclusions}

The virtualization of base stations and the design of O-RAN systems lies at the forefront of research in mobile networks. A milestone in this roadmap is finding scheduling policies that maximize the performance of vBSs while restraining their energy consumption. These policies should be practical, i.e., have minimal overheads, and applicable in realistic scenarios, meaning they should not require strong assumptions about the (often volatile) network conditions and/or user needs. The proposed learning scheme is O-RAN-compatible, has robust performance guarantees, offers a knob for prioritizing energy cost reduction, and has lightweight implementation while outperforming other computation-demanding policies (based on different learning approaches). Indeed, our extensive data-driven experiments showcase gains w.r.t. these state-of-the-art competitors that range from $35.5\%$ up to $74.3\%$ in terms of energy savings. To encourage future study in this field, we have made the source code used in this work publicly available. The significance of these results can be understood by considering the number of base stations already deployed, which is only expected to increase in the near future. Finally, our analysis identifies exciting directions for future work, such as improving the learning bounds by reducing further the dependency on the policy dimension. 

\vspace{0.1cm}
\section{Acknowledgments} \label{sec:acknowledgments}

This work was supported by the European Commission through Grant No. 101017109 (DAEMON).
\vspace{0.1cm}
\bibliography{references.bib}
\bibliographystyle{IEEEtran}
\end{document}